\documentclass[12pt]{article}

\usepackage{amsmath}
\usepackage[round]{natbib}
\usepackage{amssymb,color}
\usepackage{appendix}
\usepackage{amsthm}
\usepackage{enumerate}
\usepackage{hyperref}
\usepackage{mathtools} % need for `show only references'
\mathtoolsset{showonlyrefs=true}  % only equations which are labeled AND referenced will be numbered.                                                                                                 % IMPORTANT NOTE...must use \eqref{} instead of (\ref{})

\usepackage[margin=0.8in]{geometry}

\numberwithin{equation}{section}

\usepackage{titling}
\newtheorem{thm}{Theorem}[section]

\newtheorem{defi}{Definition}[section]
\newtheorem{assume}{Assumption}
\newtheorem{prop}{Propositioin}[section]

\newtheorem{lemma}{Lemma}[section]

\newtheorem{remark}{Remark}[section]

\begin{document}

\author{Hyungbin Park \thanks{hyungbin@snu.ac.kr, hyungbin2015@gmail.com} \\ \\ \normalsize{Department of Mathematical Sciences,} \\ 
	\normalsize{Seoul National University,}\\
	\normalsize{1 Gwanak-ro, Gwanak-gu, Seoul, Republic of Korea} 
}
\title{A representative agent model based on  risk-neutral prices}

\maketitle

\abstract{

In this paper, we determine a representative agent model based on  risk-neutral information.
 The main idea is that the pricing kernel is transition independent, which is supported by the
well-known capital asset pricing theory. 
Determining  the representative agent model is closely
related to the eigenpair problem of a second-order differential operator. The purpose of
this paper is to find all such eigenpairs which are financially or economically meaningful. We provide a necessary and sufficient condition for the existence
of such pairs, and 
prove that that all the possible eignepairs can
be expressed by a one-parameter family.
Finally, we find a representative agent model derived from the eigenpairs.
}

\section{Introduction}

A representative agent model is an important concept in finance, and many authors used a utility function of the representative agent to solve many finance problems.
However, the choice of the utility function 
 has been a result of   coincidence
rather than of consequence. 
Many authors have used 
the power functions, logarithm functions, exponential functions and their linear combinations  as a prototype of utility function. The reason why they choose one of these is the analytic tractability rather than inevitable economic foundation.  
It would be great if one can find an inevitable reason why we should choose a specific function, and this paper partially answers this.

Our story begins with the pricing kernel. 
In finance, many authors made efforts to illuminate the relation between the return and risk of assets. 
This relation is reflected in a pricing kernel, which  is determined by the interaction of risk preference of market agents. The representative agent theorem says that the interaction can be understood by the utility function of the market representative agent.
The main purpose is to express the utility function of the representative agent.
This paper shares the same idea with the Ross recovery theorem (\cite{carr2012risk}, \cite{park2016ross}, \cite{qin2014positive}, \cite{ross2015recovery},  \cite{walden2017recovery}). They utilize the risk-neutral information which can be obtained from derivative prices to find the pricing kernels. This result is attractive
because they overturned the common belief that the objective measure cannot be inferred from the derivative prices.
Later, we will conclude that the utility function  can be determined from the risk-neutral information as a one-parameter family.

This paper gives an interesting result. Later, in Section  \ref{sec:BS}, we will see that if the stock price follows the standard Black-Scholes model,
then the utility function is determined as the power function.
It is an amazing result that the Black-Scholes model induces the power utility function because two theories have  been  developed independently.

The macroeconomic foundation of determining  utility functions    relies on 
the continuous-time consumption-based {\em capital asset pricing model} (CAPM). The standard argument of the  CAPM  says that the reciprocal of the pricing kernel 
can be expressed as
\begin{equation*}
L_{t}=e^{\beta t}\,\frac{U'(X_{0})}{U'(X_{t})}\,,\;t\ge0
\end{equation*}
where $U(\cdot)$ is the representative agent, $\beta>0$ is the discount rate of the agent and $(X_t)_{t\ge0}$ is the aggregate consumption (or income) process. 
The reader can find out that this expression is a continuous-time version of Eq.(9) in \cite{ross2015recovery}.
By defining 
\begin{equation}\label{eqn:relationship_phi_U}
\phi(x):=\frac{U'(X_{0})}{U'(x)}\;,
\end{equation}
we obtain the transition independent form 
$$L_t=e^{\beta t}\phi(X_t)\,,\;t\ge0$$
stated in Assumption
\ref{assume:transition_indep} later.

The properties of the function $\phi$ are inherited  from the utility function $U$ in the CAPM.
The usual conditions on   utility functions are stated in Eq.\eqref{eqn:usual_cond_on_U}.
Therefore, from the relation in Eq.\eqref{eqn:relationship_phi_U}, the function   $\phi$ 
satisfies:
(i) $\phi>0,$ (ii) $\phi'>0,$ (iii) $\lim_{x\rightarrow0+}\phi(x)=0,$ (iv) $\lim_{x\rightarrow\infty}\phi(x)=\infty\,.$
These four conditions on $\phi$ will be also called by the {\em usual conditions}.
This issue will be discussed in Section \ref{sec:usual}.

One of the most important observations is that the problem of finding the pair 
$(\beta,\phi)$ in the transition independent form 
is transformed into  an eigenvalue/eigenfunction 
problem of an operator. 
Motivated by the paper of \cite{carr2012risk}, we will see that $(\beta,\phi)$ satisfies a second-order differential equation
\begin{equation} \label{eqn:DE}
\frac{1}{2}\sigma^2(x){\phi ''(x)}+k(x)\phi '(x)
-r(x)\phi (x) =-\beta \,\phi (x)\,.
\end{equation}
This differential equation is a continuous version of the Perron-Frobenius type equation given at Eq.(15) in \cite{ross2015recovery}.

One of the main purposes of the present paper is to find the solution pairs of Eq.\eqref{eqn:DE} with  $\phi$ satisfying the usual conditions.
The set of all such pairs is called the {\em usual set} and is denoted by $\mathcal{U}.$ 
The usual set is closely related to the behavior of the underlying process $X,$ especially the Feller boundary classification of  $X.$ 
In Section \ref{sec:usual}, the usual set $\mathcal{U}$  will be characterized in the terms of the boundary classifications of $X.$

There are two main contributions in the paper. First,
we provide a necessary and sufficient condition for the existence of such pairs  in the terms of the boundary classifications of $X.$ 
Moreover, we will offer a representation of all such pairs $(\beta,\phi).$ 
It will be proven that the possible set can be expressed by a one-parameter family with upper bounded parameter.
Second, our result 
gives an interesting implication on calibrating the objective dynamics.
By now, we have used non-parametric or heuristic parametric models to calibrate the dynamics of $X_t$ under objective measure.
The results of this paper tells us that one can use parametric models, which are economically derived, to calibrate the objective dynamics. 

The rest of this paper is structured as follows.
Section \ref{sec:finan_market} describes the underlying market model of this paper and define objective measures.
Section \ref{sec:analy} gives the main result of this paper. We transform the problem of determining the utility functions as eigenpair problems of a second-order differential operator. And then, financially meaningful eigenpairs are chosen. Section \ref{sec:recovery_thm} applies the 
 main result to find a representative agent model.
Section \ref{sec:ex} present several examples and the last section summarizes this paper. All proofs are in appendices.

%\subsection{Our plan}\begin{enumerate}	\item Begin with the candidate tuples	\item Restirct the ``divegence to infinity" condition	\item 	\item Restrict the usual condition 	\item Finally we arrive ...\end{enumerate}

\section{Financial market models}
\label{sec:finan_market}

\subsection{Risk-neutral markets}
A {\em risk-neutral} financial market is defined as a probability space
$(\Omega,(\mathcal{F}_{t})_{t\ge0},\mathbb{Q})$ having a one-dimensional Brownian motion $W=(W_t)_{t\ge0}$
with the usual completed filtration $(\mathcal{F}_{t})_{t\ge0}$ generated by $W.$
All the processes
in this article are assumed to be adapted to this filtration $(\mathcal{F}_t)_{t\ge0}.$ 
We fix a positive process $G=(G_t)_{t\ge0},$ which is called a {\em numeraire}.  
A process $S$ is said to be an {\em asset} if the discounted process $S/G$ is a martingale under $\mathbb{Q}.$ 
It is customary that this measure $\mathbb{Q}$ is referred to as a risk-neutral measure when $G$ is a money market account.

\begin{assume}\label{assume:X} 
	Let $b$ and $\sigma$ be two continuously differentiable functions on $(0,\infty).$ Assume that the stochastic differential equation (SDE)
	$$dX_{t}=b(X_{t})\,dt+\sigma(X_{t})\,dW_{t}\,,\;X_0=\xi$$
	has a unique strong solution and that the solution is non-explosive on $(0,\infty).$
	This process $X$ is referred to as a {\em state variable} in the market.	
\end{assume}
\noindent It is well-known that the process $X$ is
a univariate time-homogeneous Markov diffusion. 
The non-explosiveness means that     both boundaries $0$ and $\infty$ are inaccessible (Definition \ref{def:boundary_clasify}). This boundary condition is plausible because the aggregate consumption in the CAPM does not reach $0$ and $\infty$  in finite time.

\begin{assume} \label{assume:G}
	The state variable $X$ determines the 
	dynamics of the numeraire $G.$ More precisely,
	assume that there are continuously differentiable 
	  functions $r$ and $v$   on $(0,\infty)$
	such that
	$$G_t=e^{\int_0^t(r(X_s)+\frac{1}{2}v(X_s))\,ds+\int_0^tv(X_s)\,dW_s}\,.$$
	In the SDE form, 
	 the process $G$ follows
	\begin{equation}\label{eqn:G}
	\frac{dG_t}{G_t}=(r(X_t)+v^2(X_t))\,dt+v(X_t)\,dW_t\;,\; G_0=1\;.
	\end{equation} 
	Assume that   a local martingale
	$$(e^{-\frac{1}{2}\int_{0}^t v^2(X_s)\,ds-\int_0^t v(X_s)\,dW_s})_{t\ge0}$$
	is a martingale under $\mathbb{Q}.$
\end{assume}
\noindent The drift term in Eq.\eqref{eqn:G} may look unnatural, but it does not (\cite{park2016ross}).

\subsection{Objective measures}

Given a risk-neutral market  $(\Omega,(\mathcal{F}_{t})_{t\ge0},\mathbb{Q})$ satisfying A\ref{assume:X} - \ref{assume:G}, we want to find possible objective measures satisfying A\ref{assume:transition_indep} - \ref{assume:U} stated below.

\begin{assume}\label{assume:transition_indep}
	There are a real number $\beta$ and a twice continuously differentiable function $\phi$ with $\phi(\xi)=1$  such that 
		$$e^{\beta t}\phi(X_t)/G_t\,,\;t\ge0$$
is a martingale under $\mathbb{Q}.$
\end{assume}

\begin{defi} 
		Let $(\beta,\phi)$ be a pair satisfying A\ref{assume:transition_indep}. A measure $\mathbb{P}^\phi$ on each $\mathcal{F}_t$ defined by
	$$\frac{d\mathbb{P}^\phi}{d\mathbb{Q}\,}\Big|_{\mathcal{F}_t}={e^{\beta t}\phi(X_t)}/{G_t}$$
	is called the {\em transformed  measure} or the {\em objective measure} with respect to $\phi.$   
\end{defi}

\noindent This measure $\mathbb{P}^\phi$ is well-defined for all $T\ge0$ since $\mathbb{E}^\mathbb{Q}(\mathbb{I}_AM_t)=\mathbb{E}^\mathbb{Q}(\mathbb{I}_AM_T)$
holds for $A\in\mathcal{F}_t$ and $0\le t\le T$ where $M_t:={e^{\beta t}\phi(X_t)}/{G_t}.$
\begin{remark} In this definition, the pair $(\beta,
	\phi)$ determines the objective measure $\mathbb{P}^\phi.$
	The reader may wonder why we use the superscript $\mathbb{P}^\phi$ instead of $\mathbb{P}^{(\beta,\phi)}.$ Later we will see that in fact the number $\beta$ is automatically determined by $\phi,$ thus the notation $\mathbb{P}^{(\beta,\phi)}$ is abundant.
\end{remark}

\begin{remark}
	 Assumption \ref{assume:transition_indep} means that the pair $(\beta,\phi)$ can be understood as an eigenpair of an operator.
	 Define a pricing operator $\mathcal{P}_T$ by 
	 $$\mathcal{P}_Tf(x):=\mathbb{E}^\mathbb{Q}_{X_0=x}(f(X_T)/G_T)\,,$$
	 then
	 $$\mathcal{P}_T\phi(x)=\mathbb{E}^\mathbb{Q}_{X_0=x}(\phi(X_T)/G_T)=e^{-\beta T}\mathbb{E}^\mathbb{Q}_{X_0=x}(e^{\beta T}\phi(X_T)/G_T)=e^{-\beta T}\phi(x)\,.$$
	 The last equality holds because $(e^{\beta t}\phi(X_t)/G_t)_{0\le t\le T}$ is a martingale and its time-$0$ value is $\phi(x).$	 
\end{remark}

\begin{assume}\label{assume:trans_X}
	Under an objective  measure $\mathbb{P}^\phi,$ 
	the process $X=(X_t)_{t\ge0}$ approaches to infinity as $t\rightarrow\infty$ with probability $1.$  
\end{assume}

\noindent This assumption may be debatable,
but typical empirical data says  the aggregative consumption grows as time goes.
Thus, it is reasonable to assume that the aggregate consumption process is not recurrent nor converge to zero as time goes under   objective measures. 
It is well known that any transient time-homogeneous Markov diffusion process which does not converge to the left boundary always approaches to the right boundary as $t\rightarrow\infty.$ In this respect, this assumption  is plausible.

We now demonstrate reasonable assumptions on $\phi.$ 
The usual conditions on utility function $U(x)$ are as follows:
(i) $U'>0$ (ii) $U''<0$ (iii) $\lim_{x\rightarrow0+}U'(x)=\infty$ (iv) $\lim_{x\rightarrow\infty}U'(x)=0.$
From $\phi(x)=U'(X_0)/U'(x),$ we obtain the following assumptions by direct calculation.
\begin{assume}\label{assume:U}
	The function $\phi$  satisfies the following conditions
	\begin{enumerate}[(i)]
		\item $\phi>0$
		\item $\phi'>0$
		\item$\lim_{x\rightarrow0+}\phi(x)=0$
		\item $\lim_{x\rightarrow\infty}\phi(x)=\infty\,.$ 	
	\end{enumerate}
	In this paper, the four conditions above are called the {\em usual conditions} on $\phi.$ 
\end{assume}
\noindent It can be easily checked that a function $\phi$ satisfies the usual conditions if and only if $U(x):=\int_\cdot^x1/\phi(y)\,dy$  satisfies the usual conditions on utility function.

\begin{defi}
	For a real number $\beta$ and a positive function $\phi\in C^2(0,\infty),$ 
	we say   $(\beta,\phi)$ is an admissible pair if the pair satisfies A\ref{assume:transition_indep} - \ref{assume:U}. 
\end{defi}

\section{Analyzing the model}
\label{sec:analy}
The main purpose of this paper is to find all admissible pairs $(\beta,\phi)$  (i.e., satisfying  A\ref{assume:transition_indep} - \ref{assume:U})
for given risk-neutral market  $(\Omega,(\mathcal{F}_{t})_{t\ge0},\mathbb{Q})$ with A\ref{assume:X} - \ref{assume:G}.
The main contribution of this article is to investigate A\ref{assume:U}. 
The followings are of interest to us.
\begin{enumerate}[(i)]
	\item  Find sufficient conditions on  the risk-neutral market satisfying A\ref{assume:X} - \ref{assume:G} such that there exists an admissible pair.
	\item  If such  admissible pairs exist, express all such pairs.
\end{enumerate}
It will be proven that such pairs can be expressed by a one-parameter family.

\subsection{Second-order differential operator}

In this section, we observe that the problem of finding pairs $(\beta,\phi)$ satisfying A\ref{assume:transition_indep}  is transformed into a problem of finding eigenpairs of a second-order differential operator. Define an operator $\mathcal{L}$ by
\begin{equation}\label{eqn:L}
\mathcal{L}\phi(x)=\frac{1}{2}\sigma^2(x){\phi ''(x)}+k(x)\phi '(x)
-r(x)\phi (x)\,, 
\end{equation}
where 
$k:=b-\sigma v.$ 
See \cite{park2016ross} for proof of the following theorem.
\begin{thm}\label{thm:ode}
	Let $\beta$ be a real number and $\phi$ be a twice continuously differentiable function.
	If the pair $(\beta,\phi)$  satisfies Assumption \ref{assume:transition_indep}, then $(\beta,\phi)$ is an eigenpair of the operator $-\mathcal{L}.$ Conversely, if $(\beta,\phi)$ is an eigenpair of the operator $-\mathcal{L},$ then 	$$e^{\beta t}\phi(X_t)/G_t\,,\;t\ge0$$
	is a local martingale under $\mathbb{Q}.$
\end{thm}

\noindent Through this paper, a generic solution pair of the differential equation $\mathcal{L}h=-\lambda h$ will be denoted by  $(\lambda,h).$ This theorem guides the strategy of this paper. We first restrict our attention to the solution pairs $(\lambda,h)$ of   $\mathcal{L}h=-\lambda h$
with   $h>0$ and $h(X_0)=1,$ then we exclude pairs $(\lambda,h)$ which does not satisfy  A\ref{assume:transition_indep} - \ref{assume:U}.

\begin{thm}\label{thm:positive_r}
	If $r(x)\geq0,$  
	there exists a number $\overline{\lambda}\geq0.$  such that it
	has two linearly independent positive solutions for $\lambda<\overline{\lambda},$ has no positive solution for
	$\lambda>\overline{\lambda}$ and has one or two linearly independent solutions for $\lambda=\overline{\lambda}.$  
\end{thm}
\noindent Define $\overline{r}:=\inf_{x>0}r(x).$ By considering $r(x)-\overline{r},$ it is obtained that
\begin{equation}\label{eqn:lower_bound}
\overline{\lambda}\geq\overline{r}\;.
\end{equation}
Refer to page 146 and 149 in \cite{pinsky1995positive} for proof. 
The condition that  $r(x)\geq0$  is financially reasonable because the short interest rate $r(X_t)$ is usually nonnegative in practice.

We express the pair $(\lambda,h)$ more efficiently by the following way. 
A solution $h$ of a second-order differential equation is uniquely determined by the initial value $h(X_0)$ 
and the initial derivative $h'(X_0).$
By normalizing, we may assume $h(X_0)=1,$ then  a solution is determined by $h'(X_0).$
Thus, a solution pair $(\lambda,h)$ can be represented by $(\lambda,h'(X_0))$ under the assumption that $h(X_0)=1.$
Occasionally we use the terminology without ambiguity:
the {\em tuple}
$(\lambda,h'(X_0))$ is corresponding to the  {\em pair} $(\lambda,h).$
The two terms tuple and pair will be used to distinguish between these meanings.
Using the notion of tuples, we define the set of all admissible tuples by the following way.
$$\mathcal{A}:=\left.\left\{(\lambda,h'(\xi))\in\mathbb{R}^{2}\right| (\lambda,h) \textnormal{ satisfies Assumption  \ref{assume:transition_indep}, \ref{assume:trans_X} and \ref{assume:U}}  \right\}\,.$$

Motivated by Theorem \ref{thm:ode}, we first consider the set of  all solution pairs $(\lambda,h)$ of $\mathcal{L}h=-\lambda h$ with $h>0.$ We recall that $X_0=\xi$ in Assumption \ref{assume:X}. 

\begin{defi} 
	We say $(\lambda,h'(\xi))\in \mathbb{R}^{2}$ is a candidate tuple or we say
	$(\lambda,h)$ is a candidate pair
	if
	$(\lambda,h)$ is a solution pair of $\mathcal{L}h=-\lambda h$ with
	$h(\cdot)>0$ and $h(\xi)=1.$ 
	Denote the set of the candidate tuples by $\mathcal{C}.$
	$$\mathcal{C}:=\left.\left\{(\lambda,h'(\xi))\in\mathbb{R}^{2}\right| \mathcal{L}h=-\lambda 
	h,\, h>0,\, h(\xi)=1 \right\}\;.$$
\end{defi}  
\noindent We briefly state properties
of $\mathcal{C}.$
Theorem \ref{thm:positive_r} guarantees that the set $\mathcal{C}$ is nonempty if  $r(x)\geq0.$ 
Recall that $\overline{\lambda}$ be the maximum value of the first coordinate of elements of 
$\mathcal{C},$ that is, $\overline{\lambda}=\max\{\,\lambda \,|\,(\lambda,z)\in\mathcal{C} \,\}.$
For any $\lambda$ with $\lambda\leq \overline{\lambda},$
we set
$$M_{\lambda}:=\sup_{(\lambda,z)\in\mathcal{C}}z\,,\;\;
m_{\lambda}:=\inf_{(\lambda,z)\in\mathcal{C}}z\;.$$
Occasionally, we use notations $M(\lambda)$ and $m(\lambda)$ instead of $M_\lambda$ and $m_\lambda,$ respectively, to avoid double subscripts such as $h_{M_\lambda}$ and $h_{m_\lambda}.$

\begin{prop} \label{prop:slice}
	Let $\lambda\leq\overline{\lambda}.$
	For any $z$ with $m_{\lambda}\leq z\leq M_{\lambda},$ the tuple $(\lambda,z)$ is in 
	$\mathcal{C}.$
\end{prop}
\noindent The above proposition can be easily shown by the fact that the equation $\mathcal{L}h=-\lambda h$ has two linearly independent solutions and any solution can be expressed by the linear combinations of the two solutions.
For rigorous proof, refer to \cite{park2016ross}.
Therefore,  the $\lambda$-slice of $\mathcal{C}$ is a connected and compact set.

\subsection{Transformed measures}
\label{sec:transformed}
For a candidate pair $(\lambda,h),$
Theorem \ref{thm:ode} says that  
$(e^{\lambda t}\,h(X_t)/G_t)_{t\ge0}$ 
is a local martingale under $\mathbb{Q}.$

\begin{prop} \label{prop:dynamics_under_P}
		Let $(\lambda,h)$ be a candidate pair such that
	$(e^{\lambda t}\,h(X_t)/G_t)_{t\ge0}$ is a martingale (that is, $(\lambda,h)$ satisfies Assumption \ref{assume:transition_indep}). Then 	
	a process $(B_t^{h})_{t\ge0}$ defined by
	$dB_{t}^{h}=-(\sigma h'h^{-1}-v)(X_{t})\,dt+dW_{t}$
	is a Brownian motion under the transformed measure $\mathbb{P}^h.$
	In this case,   the $\mathbb{P}^h$-dynamics of $X$ is
	\begin{equation}\label{eqn:X_under_P}
	\begin{aligned}
	dX_{t}&=(b-v\sigma+\sigma^{2}h'h^{-1})(X_{t})\,dt+\sigma(X_{t})\,dB_{t}^{h}\\
	&=(k+\sigma^{2}h'h^{-1})(X_{t})\,dt+\sigma(X_{t})\,dB_{t}^{h}\,.
	\end{aligned}
	\end{equation}
\end{prop}
\noindent Occasionally, we use the notation $\mathbb{P}$ and $(B_t)_{t\ge0}$ instead of $\mathbb{P}^h$ and $(B_t^h)_{t\ge0},$ respectively, without ambiguity.
Even when $(e^{\lambda t}\,h(X_t)/G_t)_{t\ge0}$ is not a martingale, we can consider the diffusion process in Eq.\eqref{eqn:X_under_P}.
\begin{defi}\label{defi:induced_diffusion}
	The diffusion process $(X_t)_{t\ge0}$ defined by 
	$$dX_t=(k+\sigma^2 h'/h)(X_t)\,dt+\sigma(X_t)\,dB_t$$
	is called the diffusion process {\em induced by} the pair $(\lambda,h)$ or the tuple $(\lambda,h'(\xi))$. 
\end{defi}

\subsection{Divergence to infinity}
We shift our attention to Assumption \ref{assume:trans_X}.
The following theorem specifies which candidate tuples induce transformed measures satisfying Assumption \ref{assume:trans_X}.
For proof, see \cite{park2016ross}.
\begin{thm} %\label{thm:main_thm}
	Let $\lambda\leq\overline{\lambda}.$	
	The diffusion process 
	induced by tuple $(\lambda,M_{\lambda})$
	approaches to infinity as $t\rightarrow\infty$ with probability one.
	For $z$ with $m_{\lambda}\leq z<M_{\lambda},$
	the diffusion process induced by tuple $(\lambda,z)$
	approaches to zero as $t\rightarrow\infty$ with positive probability.
\end{thm}

\noindent In conclusion, it is obtained that
$$\mathcal{A}\subseteq\{(\lambda,M_\lambda)\in\mathcal{C}\,|\,\lambda\leq\overline{\lambda}\,\}\;.$$
From now, we mainly focus on the tuples $(\lambda,M_\lambda)$ with $\lambda\leq\overline{\lambda}.$

\subsection{The martingale condition}
\label{sec:mart_condi}
We now explore the martingale condition  discussed in Section \ref{sec:transformed}. 
For any given tuple $(\lambda,h'(\xi))$, we know the process $e^{\lambda t}\,h(X_t)\,G_t^{-1}$ is a local martingale.
Consider the set of tuples $(\lambda,M_\lambda)$ which induce the martingales $e^{\lambda t}\,h(X_t)\,G_t^{-1},$ that is,
$$\mathcal{M}:=\{\,(\lambda,M_\lambda)\,|\,e^{\lambda t}\,h(X_{t})\,G_t^{-1} \textnormal{ is a martingale}\,\}\;.$$
Clearly, $\mathcal{A}$ is a subset of $\mathcal{M}.$
The following theorem states that the set is connected in $\mathbb{R}^2.$ Refer to \cite{park2016ross} for proof.

\begin{thm} \label{thm:monotone_martingality}
	Let $\delta<\lambda\leq\overline{\lambda}$ and let 
	$(\delta,g)$ and $(\lambda,h)$ be the candidate pairs corresponding to $(\delta,M_\delta)$ and $(\lambda,M_\lambda),$ respectively. If $e^{\delta t}\,g(X_{t})\,G_t^{-1}$ is a martingale, so is $e^{\lambda t}\,h(X_{t})\,G_t^{-1}.$
\end{thm}

We  assume that sufficiently many candidate pairs satisfy the martingale condition. In other words, the number $\lambda_0$ defined by
$\lambda_0:=\inf\{\,\lambda\,|\,(\lambda,M_\lambda)\in\mathcal{M}\,\}$
is sufficiently small. 
This assumption is to guarantee the existence of admissible pair.
If the set is too small or is empty, there may not exist an admissible pair.

There is an useful criteria to check the martingale condition. Let $(\lambda,h)$ be a candidate pair of $\mathcal{L}h=-\lambda h.$ This pair satisfies the martingale condition if and only if the following two conditions hold:
\begin{equation}\label{eqn:criteria}
\begin{aligned}
&\int_{0}^{\xi}\!dx\,\frac{1}{h^2(x)}e^{-\int_{\xi}^{x}\frac{2k(s)}{\sigma^2(s)}\,ds}\int_{x}^{\xi}\!dy\,\frac{h^2(y)}{\sigma^2(y)}e^{\int_{\xi}^{y}\frac{2k(s)}{\sigma^2(s)}\,ds}=\infty\;,\\
&\int_{\xi}^{\infty}\!dx\,\frac{1}{h^2(x)}e^{-\int_{\xi}^{x}\frac{2k(s)}{\sigma^2(s)}\,ds}\int_{\xi}^{x}\!dy\,\frac{h^2(y)}{\sigma^2(y)}e^{\int_{\xi}^{y}\frac{2k(s)}{\sigma^2(s)}\,ds}=\infty\;.
\end{aligned} 
\end{equation}
We recall   Definition \ref{defi:induced_diffusion}.
The above criteria implies that 
the process $e^{\lambda t}\,h(X_{t})\,G_t^{-1}$ is a martingale if and only if the diffusion process induced by $(\lambda,h)$  does not explode. 
Refer to page 215 in \cite{pinsky1995positive}.
The martingale condition can be checked case-by-case, thus we do not go further details.

\subsection{The usual conditions}
\label{sec:usual}
One of the main contributions of the present article is to investigate Assumption \ref{assume:U}.
Now, in the set $\{(\lambda,M_\lambda)\in\mathcal{C}\,|\,\lambda\leq\overline{\lambda}\,\},$ we explore which  tuples  satisfy
Assumption \ref{assume:U}.
For convenience, put
$$\mathcal{U}:=\{(\lambda,M(\lambda))\in\mathcal{C}\,|\,h_{M(\lambda)} \textnormal{ satiafies the usual conditions} \}\;.$$
Here, $h_{M(\lambda)}$ is the function corresponding to the tuple $(\lambda,M_\lambda).$
The notation $\mathcal{U}$ is inherited from terminology ``usual conditions".
Let $\mathbb{L}$ be the measure   defined by the Radon-Nikodym derivative
$$\left.\frac{d\mathbb{L}}{d\mathbb{Q}}\right|_{\mathcal{F}_t}=\exp{\left(-\frac{1}{2}\int_{0}^t v^2(X_s)\,ds-\int_0^t v(X_s)\,dW_s\right)}\;,$$
which is a martingale by Assumption \ref{assume:G}.
The $\mathbb{L}$-dynamics of $X_t$ is
\begin{equation*} 
\begin{aligned}
dX_{t}&=(b-v\sigma)(X_{t})\,dt+\sigma(X_{t})\,dB_{t}^{1}\\
&=k(X_{t})\,dt+\sigma(X_{t})\,dB_{t}^{1} 
\end{aligned}
\end{equation*}
for a Brownian motion $B_t^1.$ Here, we used notation $B_t^1$  to be consistent with the notation used in Proposition \ref{prop:dynamics_under_P}.

To investigate the set $\mathcal{U},$ we need to employ the notion of boundary classification.
Both boundaries $0$ and $\infty$ of $(X_t)_{t\ge0}$ are inaccessible under $\mathbb{Q}$ if and only if those are inaccessible under $\mathbb{L}.$ It is because two measures $\mathbb{Q}$ and $\mathbb{L}$ are equivalent on each $\mathcal{F}_T,T\ge0.$ 
From Assumption \ref{assume:X}, two boundaries $0$ and $\infty$ are inaccessible  under both measures $\mathbb{Q}$ and $\mathbb{L}.$
%The boundary $0$   is a natural (or entrance) boundary of $X_t$ under $\mathbb{Q}$ if and only if  $0$ is  a natural (or entrance) boundary of $X_t$ under $\mathbb{L}.$ The same statement is true for the boundary $\infty.$ Therefore, we do not distinguish the boundary classification of $X_t$ under $\mathbb{Q}$ from that of $X_t$ under $\mathbb{L}.$

From now on, we discuss more detailed boundary classification under the measure $\mathbb{L}.$

\begin{defi} \label{def:boundary_clasify} Let
	\begin{align*}
	\gamma(x)&=e^{-\int_{\xi}^x\frac{2k(s)}{\sigma^2(s)}\,ds}\;,\\
	Q(x)&=\frac{2}{\sigma^2(x)\gamma(x)} \int_{\xi}^x\gamma(s)\,ds\;,\\
	R(x)&=\gamma(x)\int_{\xi}^x\frac{2}{\sigma^2(s)\gamma(s)} \,ds\;.
	\end{align*} 
	An endpoint $0$ is said to be {\em inaccessible} if $R\notin L^1(0,\xi).$ An inaccessible endpoint $0$ is said to be
	\begin{align*}
	\left\{\enspace
	\begin{aligned}
	&\textnormal{entrance if}&&Q\in L^1(0,\xi) \;,\\
	&\textnormal{natural if} &&Q\notin L^1(0,\xi) \;.
	\end{aligned}
	\right.
	\end{align*} 
	The definitions of inaccessible, entrance and natural at the endpoint $\infty$ are defined in similar ways. 
\end{defi}

We now state main theorems of the paper, which describes the usual set $\mathcal{U}.$
The following theorem implies that the set $\mathcal{U}$ is a connected subset of $\mathbb{R}^2.$ Define 
$\lambda_1:=\sup\{\,\lambda\,|\,(\lambda,M_\lambda)\in\mathcal{U}\,\},$ then   for all $\lambda<\lambda_1,$ the tuple 
$(\lambda,M_\lambda)$ is in the usual set $\mathcal{U}.$
The endpoint   $(\overline{\lambda},M_{\overline{\lambda}})$
may or may not be in $\mathcal{U}.$

\begin{thm} \label{thm:monotone_usual}
	Assume $\delta<\lambda\leq\overline{\lambda}.$
	Let $g$ and $h$ be the functions corresponding to tuple $(\delta,M_\delta)$ and $(\lambda,M_\lambda),$ respectively.  If $h$ satisfies the usual conditions, then so does $g.$ In other words, if $(\lambda,M_\lambda)$ is in $\mathcal{U},$ then   $(\delta,M_\delta)$ is also in $\mathcal{U}.$
\end{thm}

\begin{thm}\label{thm:usual_condi_iff_condi}
	Assume $r(\cdot)\geq 0$ and  $r(\cdot)$ is bounded on $(0,\xi).$ Then the set $\,\mathcal{U}$ is nonempty if and only if  $0$ is a natural boundary.
	In this case, for $$\lambda<\overline{r}:=\inf_{x>0} r(x)\,,$$  the tuple $(\lambda,M_\lambda)$ is in the set $\,\mathcal{U}.$
\end{thm}
\noindent The above theorem states a sufficient and necessary condition for the existence of   tuples  which   satisfies
the usual conditions. 
Refer to Appendix \ref{app:pf_thm_II} and \ref{app:pf_thm_I}  for proofs of Theorem  \ref{thm:monotone_usual} and \ref{thm:usual_condi_iff_condi}, respectively.

For the remainder of this section, we find the usual set $\mathcal{U}$ when the short interest rate function $r(\cdot)$ is a constant $r.$
By Theorem \ref{thm:usual_condi_iff_condi}, $\mathcal{U}$ is nonempty if and only if $0$ is a natural boundary, thus we assume $0$ is a natural boundary. For $\lambda<r,$ the tuple $(\lambda,M_\lambda)$ is always in $\mathcal{U}.$

From equation \eqref{eqn:lower_bound}, we know that $\overline{\lambda}\geq r.$
The case of $\overline{\lambda}=r$ is relatively easy to find the set $\mathcal{U}.$ Since
$$\{(\lambda,M_\lambda)\,|\,\lambda<r\}\subseteq\mathcal{U}\subseteq\{(\lambda,M_\lambda)\,|\,\lambda\leq r\}\;,$$
The set  $\mathcal{U}$ is determined by the solution corresponding to the tuple $(r,M_r).$
Consider the solution of the corresponding second-order differential equation
\begin{equation*} 
\frac{1}{2}\sigma^2(x)h''(x)+k(x)h'(x)=0\;. 
\end{equation*}
By direct calculation, two linearly independent solutions are
$$h_1(x)=1+c\int_{\xi}^{x}e^{-\int_{\xi}^{y}\frac{2k(s)}{\sigma^2(s)}\,ds}\,dy\;,\;\;h_2(x)=1\;.$$
Clearly, $h_2(x)=1$ does not satisfy the usual conditions.
By considering the function $h_1(x),$ we have the following proposition. Recall that  $$\gamma(x):=e^{-\int_{\xi}^{x}\frac{2k(s)}{\sigma^2(s)}\,ds}\;.$$
\begin{prop} \label{prop:U_when_overline_lambda_r}
	Assume that $\overline{\lambda}=r$ and $0$ is a natural boundary.
	If $ \int_{\xi}^{\infty}\gamma(x)\,dx=\infty$ and $ \int_{0}^{\xi}\gamma(x)\,dx<\infty,$ then 
	$\mathcal{U}=\{(\lambda,M_\lambda)\,|\,\lambda\leq r\}.$
	Otherwise, $\mathcal{U}=\{(\lambda,M_\lambda)\,|\,\lambda<r\}.$
\end{prop}
\begin{proof} Assume
	$ \int_{\xi}^{\infty}\gamma(x)\,dx=\infty$ and $ \int_{0}^{\xi}\gamma(x)\,dx<\infty.$ Then   $h_1(x)$ with 
	$$c=\frac{1}{\int_{0}^{\xi}\gamma(x)\,dx}$$ is the function corresponding to $(r,M_r).$ Clearly $\lim_{x\rightarrow0+}h_1(x)=0$ with this choice of $c.$ Thus, this tuple is in $\mathcal{U}.$	 The converse is trivial.
\end{proof}
\noindent It is noteworthy that the conditions $ \int_{\xi}^{\infty}\gamma(x)\,dx=\infty$ and $ \int_{0}^{\xi}\gamma(x)\,dx<\infty$ means that the diffusion process $X_t$ under $\mathbb{L}$
has the following property:
$$\mathbb{L}\Big(\lim_{t\rightarrow\infty}X_t=0\Big)=\mathbb{L}\Big(\sup_{0\leq t<\infty}X_t<\infty\Big)=1\;.$$
Refer to page 345 in \cite{karatzas2012brownian}.

We now consider the case of $\overline{\lambda}>r.$ Refer to Appendix \ref{app:pf_thm_const_r} for proof of the following theorem.

\begin{thm}\label{thm:usual_condi_const_r}
	Assume that $\overline{\lambda}>r$ and $0$ is a natural boundary. If $\int_{\xi}^{\infty}\gamma(x)\,dx=\infty,$ then 
	$$\mathcal{U}=\{(\lambda,M_\lambda)\,|\,\lambda\leq \overline{\lambda}\,\}\;\;\text{ or }\;\;\{(\lambda,M_\lambda)\,|\,\lambda< \overline{\lambda}\,\}\;.$$
	Moreover, if $\infty$ is a natural boundary, then  $\mathcal{U}=\{(\lambda,M_\lambda)\,|\,\lambda\leq \overline{\lambda}\,\}.$
	If  $\int_{\xi}^{\infty}\gamma(x)\,dx<\infty,$ then 
	$\mathcal{U}=\{(\lambda,M_\lambda)\,|\,\lambda< r\,\}.$ \end{thm}
\noindent The authors conjecture that when $\int_{\xi}^{\infty}\gamma(x)\,dx=\infty,$
the right boundary $\infty$ is a natural boundary if and only if  $\mathcal{U}=\{(\lambda,M_\lambda)\,|\,\lambda\leq \overline{\lambda}\,\}.$

\subsection{Admissible sets}
\label{sec:admissible}
The purpose of this paper is to find the admissible set $\mathcal{A}$ under Assumption \ref{assume:X} - \ref{assume:U}.
The admissible set satisfies
$$\mathcal{A}=\mathcal{M}\cap\mathcal{U}\;, $$
thus $\mathcal{A}$ is a connected subset because  $\mathcal{M}$ and $\mathcal{U}$ are connected  
subsets of $  \{(\lambda,M_\lambda)\in\mathcal{C}\,|\,\lambda\leq\overline{\lambda}\,\}.$ Define $$\lambda_0:=\inf\{\,\lambda\,|\,(\lambda,M_\lambda)\in\mathcal{M}\,\}\;,\;\lambda_1:=\sup\{\,\lambda\,|\,(\lambda,M_\lambda)\in\mathcal{U}\,\}\;.$$
The endpoints $\lambda_0$ and $\lambda_1$ may or may not be in $\mathcal{M}$ and $\mathcal{U},$ respectively. Assuming $\lambda_0\leq\lambda_1,$ we obtain that
for $\lambda$ between $\lambda_0$ and $\lambda_1,$  the tuple
$(\lambda,M_\lambda)$ is an admissible tuple.
In conclusion, one can recover the objective measures by the one-parameter family.

\section{The recovery theorem}
\label{sec:recovery_thm}
We investigate how the previous results can be used for Ross recovery. 
In the continuous-time consumption-based CAPM, the state variable $X_t$ is the aggregate consumption (or income) process of the market.  
In a financial market, the aggregate income is equal to the 
aggregate dividend, thus we may assume that $X_t$ is the aggregate dividend. 
Let $S_t$ be a composite stock price index such as S\&P 500 and assume that $S_t$ pays aggregate dividend which is a function of $S_t,$ that is,  
$$X_t=\delta(S_t)S_t$$
where $\delta(S_t)$ is the dividend per one unit of the composite stock price index. 
The function $\delta(s)$ is assumed to be known ex ante and is a nondecreasing function of $s.$
Assume that the function $\pi(s):=\delta(s)s$ is continuously twice differentiable with 
continuously twice differentiable inverse.
By Appendix \ref{sec:inv_prop}, the transformed measure is invariant under the map $\pi,$ thus 
we may assume that the state variable is $S_t.$
\cite{ross2015recovery}  also used the dividend or 
the composite stock price index (S\&P 500) 
in page 630-633 as the state variable.

Let the numeraire $G_t$ be the wealth process induced the composite stock price process $S_t,$ that is, $G_t=e^{\int_0^t\delta(S_u)\,du}S_t.$ 
Assume that the state variable $S_t$ satisfies
$$dS_t=(r(S_t)-\delta(S_t)+\sigma^2(S_t))S_t\,dt+\sigma(S_t)S_t\,dW_t\;.$$
Then 
$$\frac{dG_t}{G_t}=(r(S_t)+\sigma^2(S_t))\,dt+\sigma(S_t)\,dW_t\;.$$ 
The operator $\mathcal{L}$ corresponding to equation \eqref{eqn:L} is
\begin{equation}\label{eqn:2nd_ode}
\mathcal{L}h(s)=\frac{1}{2}\sigma^2(s)s^2h''(s)+(r(s)-\delta(s))sh'(s)-r(s)h(s)\;. 
\end{equation}

Occasionally, $Y_t:=\log S_t$ induces a simpler second-order equation.  
Let $y=\ln s$ and define $\kappa(y)=r(s)-\delta(s),$ $\nu(y)=\sigma(s)$ and $\rho(y)=r(s).$ Then
\begin{align*}
dY_t
&=(r(S_t)-\delta(S_t)+\frac{1}{2}\sigma^2(S_t)) \,dt+\sigma(S_t) \,dW_t \\
&=(\kappa(Y_t)+\frac{1}{2}\nu^2(Y_t))\,dt+\nu(Y_t)\,dW_t\;.
\end{align*}
The corresponding equation \eqref{eqn:2nd_ode} becomes
$$\frac{1}{2}\nu^2(y)g''(y)+(\kappa(y)-\frac{1}{2}\nu^2(y))g'(y)-\rho(y)g(y)=-\lambda g(y)$$
where $g(y)=h(s).$

\section{Examples}
\label{sec:ex}
In this section, we explore examples of Ross recovery. 
Denote by $S_t$ the  composite stock price index as discussed in Section \ref{sec:recovery_thm}. Occasionally, for convenience, we say $S_t$ is the stock price without ambiguity.
The classical Black-Scholes stock model and the exponential CIR stock model are discussed with constant short interest rate and constant dividend rate in Section  \ref{sec:BS} and \ref{sec:expo_CIR}, respectively.
The classical Black-Scholes stock model with log dividend rate  is explored in Section \ref{sec:log_dividend}.

\subsection{The Black-Scholes model}
\label{sec:BS}
The classical Black-Scholes stock model   with constant short interest rate and constant dividend rate is discussed.
The dividends of the stock are paid out continuously with
rate $\delta\,dt.$
Suppose $S_t$ follows a geometric Brownian motion
$$dS_t=(r-\delta+\sigma^2)S_t\,dt+\sigma S_t\,dW_t\;,\;S_0=1$$
and the numeraire is $G_t=S_te^{\delta t}.$ 
The corresponding second-order equation is
\begin{equation*} 
\mathcal{L}h(s)=\frac{1}{2}\sigma^2s^2h''(s)+(r-\delta)sh'(s)-rh(s)=-\lambda h(s)\;. 
\end{equation*}
By direct calculation, we have $\overline{\lambda}=\frac{1}{2}(\frac{\sigma}{2}-\frac{r-\delta}{\sigma})^2+r.$
For $\lambda\leq\overline{\lambda},$
it can be easily shown that the value $M_\lambda$ is
$$M_\lambda=\frac{1}{2}-\frac{r-\delta}{\sigma^2}+\sqrt{\left(\frac{1}{2}-\frac{r-\delta}{\sigma^2}\right)^{2}+\frac{2(r-\lambda)}{\sigma^2}}$$ and the function  corresponding to the tuple $(\lambda,M_\lambda)$ is 
$$h_\lambda(s):=s^{\frac{1}{2}-\frac{r-\delta}{\sigma^2}+\sqrt{\left(\frac{1}{2}-\frac{r-\delta}{\sigma^2}\right)^{2}+\frac{2(r-\lambda)}{\sigma^2}}}\;.$$
The function $\gamma(s)$ is
$ s^{-\frac{2(r-\delta)}{\sigma^2}}.$

We find the admissible set $\mathcal{A}.$
It can be easily checked that every candidate pair is admissible by using the method in Section \ref{sec:mart_condi}, thus $\mathcal{A}=\mathcal{U}.$
As is well-known, both endpoints $0$ and $\infty$ of the geometric Brownian motion are natural boundaries.
Applying Theorem \ref{thm:usual_condi_iff_condi}, Proposition \ref{prop:U_when_overline_lambda_r}  and Theorem \ref{thm:usual_condi_const_r}, we obtain
\begin{equation*}
\begin{aligned}
&\mathcal{A}= \{\,(\lambda,M_\lambda)\,|\,\lambda\leq\overline{\lambda}\,\}\;&&\textnormal{ if }\;2(r-\delta)<\sigma^2\;, \\
&\mathcal{A}= \{\,(\lambda,M_\lambda)\,|\,\lambda<r\,\}\;&&\textnormal{ if }\;2(r-\delta)\geq\sigma^2\;. 
\end{aligned} 
\end{equation*}

\subsection{Exponential CIR model}
\label{sec:expo_CIR}

We explore an example of stock model with 
inaccessible entrance $0$ boundary.
By Theorem \ref{thm:usual_condi_iff_condi}, the usual set is empty, that is,
$$\mathcal{U}=\varnothing\;. $$
Even though $\mathcal{U}$ is empty, it would be interesting  to find the set $\mathcal{M}.$
Assume that the short interest rate $r$ and the dividend rate $\delta$ are constants. Put $\theta:=r-\delta.$
Let $Y_t$ be an extended CIR process given by
$$dY_t=(\theta+\frac{1}{2}\sigma^2Y_t)\,dt+\sigma\sqrt{Y_t}\,dW_t $$
with $2\theta\geq\sigma^2.$
It is well known that the left boundary $0$ and the right boundary $\infty$ are entrance and natural, respectively.  
Assume the stock price follows $S_t=e^{Y_t}$ so that
$$dS_t=(\theta+\sigma^2\ln S_t)S_t\,dt+\sigma\sqrt{\ln S_t}\,S_t\,dW_t\;.$$

The corresponding second-order differential equation is
$$\frac{1}{2}\sigma^2\ln(s)s^2h''(s)+\theta s h'(s)-rh(s)=-\lambda h(s)\;.$$
However, the process $Y_t:=\log S_t$ induces a simpler second-order equation  
$$\frac{1}{2}\sigma^2yg''(y)+(\theta-\frac{1}{2}\sigma^2y)g'(y)-rg(y)=-\lambda g(y)\;.$$
It can be easily checked that
$$g_\lambda(y):=M\left(\frac{2(r-\lambda)}{\sigma^2},\frac{2\theta}{\sigma^2},y\right)$$
is a solution corresponding to $(\lambda,M_\lambda)$ for   $\lambda\leq r=\overline{\lambda}.$ It is known that the confluent hypergeometric function $M(\alpha,\beta,y)$ is positive if and only if $\alpha\leq0.$
Refer to \cite{qin2014positive} for more details. We obtain that 
\begin{equation}\label{eqn:Y_under_P}
dY_t=\left(\theta+\frac{1}{2}\sigma^2Y_t+\sigma^2Y_t\,\frac{g_\lambda'(Y_t)}{g_\lambda(Y_t)}\right)\,dt+\sigma\sqrt{Y_t}\,dB_t
\end{equation}
where $B_t$ is a Brownian motion under the corresponding transformed measure.

It can be easily checked that 
$e^{\lambda t}\,g_\lambda(Y_{t})\,G_t^{-1}$
is a martingale.  
By considering the asymptotic behavior $M(\alpha,\beta,y)\sim e^{y}y^{\alpha-\beta}/\Gamma(\alpha)$ as $y\rightarrow\infty$ and the fact that $M'(\alpha,\beta,y)=(\alpha/\beta)M(\alpha+1,\beta+1,y),$
we obtain  as $y\rightarrow\infty,$
$$\frac{g_\lambda'(y)}{g_\lambda(y)}\sim \frac{\sigma^2}{2\theta}\;.$$
The drift of equation \ref{eqn:Y_under_P} has linear growth rate, as the CIR model, $Y_t$ 
does not explode by the criteria in equation \eqref{eqn:criteria}.
In conclusion,
the function corresponding to $(\lambda,M_\lambda)$ is
$h_\lambda(s):=g_\lambda(\ln s)$ and
\begin{align*}
\mathcal{M}
&=\left.\left\{\,\left(\lambda\,,\,\frac{h_\lambda'(S_0)}{h_\lambda(S_0)}\right)\,\right|\,\lambda\leq r\,\right\}\\
&=\left.\left\{\,\left(\lambda\,,\,\frac{r-\lambda}{\theta S_0}\cdot\frac{M\left(\frac{2(r-\lambda)}{\sigma^2}+1,\frac{2\theta}{\sigma^2}+1,\ln S_0\right)}{ M\left(\frac{2(r-\lambda)}{\sigma^2},\frac{2\theta}{\sigma^2},\ln S_0\right)}\right)\,\right|\,\lambda\leq r\,\right\} \;.
\end{align*}

\subsection{Log dividend models}
\label{sec:log_dividend}

In this section, we explore the possibility of recovering when the dividends of the stock are paid out continuously with
rate $b \log S_t\,dt.$
Suppose $S_t$ has a constant volatility $\sigma$ and the short interest rate is a constant $r.$
$$dS_t=(r+\sigma^2-b\log S_t)S_t\,dt+\sigma S_t\,dW_t\;.$$
It can be easily shown that $S_t=e^{Y_t}$ where
$$dY_t=(r+\frac{1}{2}\sigma^2-bY_t)\,dt+\sigma\,dW_t\;.$$
The corresponding second-order differential equation is
$$\frac{1}{2}\sigma^2s^2h''(s)+(r-b\log s)sh'(s)-rh(s)=-\lambda h(s)\;.$$
Substituting $s=e^y$ and $h(s)=g(y),$ it follows that
$$\frac{1}{2}\sigma^2g''(y)+(r-\frac{1}{2}\sigma^2-by)g'(y)-rg(y)=-\lambda g(y)\;.$$
One can check that for some normalizing constant $c,$
$$g_\lambda(y)=c\left(\frac{M(\frac{r-\lambda}{2b},\frac{1}{2},\frac{b}{\sigma^2}(y-\kappa)^2)}{\Gamma(\frac{1}{2}+\frac{r-\lambda}{2b})}+2(y-\kappa)\sqrt{\frac{b}{\sigma^2}}\frac{M(\frac{r-\lambda}{2b}+\frac{1}{2},\frac{3}{2},\frac{b}{\sigma^2}(y-\kappa)^2)}{\Gamma(\frac{r-\lambda}{2b})}\right)$$
is an admissible function (i.e., positive increasing solution)
with $g_\lambda(-\infty)=0,$ $g_\lambda(\infty)=\infty$ for $\lambda\leq r=\overline{\lambda}$
where $\kappa=\frac{r}{b}-\frac{\sigma^2}{2b}$ and $M(\cdot,\cdot,\cdot)$ is the confluent hypergeometric function.

The corresponding transformed measure $\mathbb{P}$ is 
$$\left.\frac{d\mathbb{P}}{d\mathbb{Q}}\right|_{\mathcal{F}_{t}}=e^{\lambda t}\,g_\lambda(Y_{t})\,G_t^{-1}=e^{\lambda t}\,g_\lambda(\log S_t)\,G_t^{-1}$$
under which the dynamics of $Y_t$ is
\begin{equation}\label{eqn:Y}
dY_t=\left(r+\frac{1}{2}\sigma^2-bY_t+\sigma^2\,\frac{g_\lambda'(Y_t)}{g_\lambda(Y_t)}\right)\,dt+\sigma\,dB_t\;.  
\end{equation}
Thus, we obtain the $\mathbb{P}$-dynamics of $S_t=e^{Y_t}.$ 

It can be easily checked that 
$e^{\lambda t}\,g_\lambda(Y_{t})\,G_t^{-1}$
is a martingale. 
By considering the asymptotic behaviors of $M(\cdot,\cdot,\cdot )$ and $M'(\cdot,\cdot,\cdot )$ as in Section \ref{sec:log_dividend},
we obtain that as $|y|\rightarrow\infty,$
$$\frac{g_\lambda'(y)}{g_\lambda(y)}\sim \frac{4b}{3\sigma^2}|y|\;.$$
Because the drift of equation \eqref{eqn:Y} has linear growth rate, by the criteria in equation \eqref{eqn:criteria}, we know 
$Y_t$ does not explode with the dynamics of $Y_t$ in equation \eqref{eqn:Y}.
In conclusion, we get
\begin{align*}
\mathcal{A}
&=\left.\left\{\,\left(\lambda\,,\,\frac{h_\lambda'(S_0)}{h_\lambda(S_0)}\right)\,\right|\,\lambda\leq r\,\right\}
\end{align*}
where $h_\lambda(s)=g_\lambda(\ln s).$

%\subsection{The Vasicek model}\label{sec:Vasicek}In this section, we investigate an example of non-constant short interest rate case. Assume that the short interest rate function is $r(s)=a\log s$ for $a\in\mathbb{R}$ and $a\neq0$ and thedividend is constant $\delta.$ The stock price $S_t$ follows\begin{align*}dS_t&=(r(S_t)-\delta+\sigma^2)S_t\,dt+\sigma S_t\,dW_t\\&=(a\ln S_t-\delta+\sigma^2)S_t\,dt+\sigma S_t\,dW_t\;.\end{align*}[[The  short interest rate function $r(s)=a\log s$ does not satisfy the assumptions of Theorem \ref{thm:usual_condi_iff_condi}, but $Y_t$ satisfies..  that is the function $y$ satisfeis the conidtion.we can calculate $\mathcal{U}$ and $\mathcal{A}$  directly.]] Define $Y_t=\ln S_t,$ then $$dY_t=(aY_t-\delta+\frac{1}{2}\sigma^2)\,dt+\sigma\,dW_t\;. $$This process $Y_t$ is an extended Ornstein-Uhlenbeck process.The corresponding second-order differential equation is$$\frac{1}{2}\sigma^2s^2h''(s)+(a\ln s-\delta)sh'(s)-a\ln(s) h(s)=-\lambda h(s)\;.$$Substituting $s=e^y$ and $h(s)=g(y),$ it follows that$$\frac{1}{2}\sigma^2g''(y)+(ay-\delta-\frac{1}{2}\sigma^2)g'(y)-ay\,g(y)=-\lambda g(y)\;.$$Refer to [[Linetsky]] We have that $$\mathcal{A}$$

\section{Conclusion}
\label{sec:conclusion}

This paper determines a representative agent model from a risk-neutral measure in a continuous-time setting. 
One of the key ideas of the argument is that
the reciprocal of the pricing kernel is expressed by the {\em transition independent} form 
\begin{equation*} 
e^{\beta t}\,\phi(X_t)
\end{equation*}
for a constant $\beta$ and a positive function $\phi.$
This form is originated from the 
continuous-time consumption-based asset pricing model, which is a well-known asset pricing theory.
Based on the theory, several conditions such as
the martingale condition, divergence to infinity and the usual conditions are assumed on  
the function $\phi$ and the underlying process $X_t.$ 
The pair $(\beta,\phi)$ satisfying these conditions was called an admissible pair.

The  main purpose of this paper is to investigate the
admissible pairs.  
A necessary and sufficient condition for the existence of admissible pairs was explored.
Moreover, we showed that, if it exists, the set of admissible pairs is expressed by a one-parameter family. The admissible set is determined by the lower bound of the martingale set $\mathcal{M}$ and  the upper bound of the usual set $\mathcal{U}.$ As a special case, when the short interest rate is a constant, the set $\mathcal{U}$ was presented.

The following extensions for future research are suggested. First, it would be interesting to extend the recovery to multi-dimensional state variables.
In this case, the corresponding Sturm-Liouville equation is a second-order partial differential equation.
Second, it would be valuable to find   economically meaningful methods to determine $\beta.$ 
We could 
not provide such methods in this article.
Finally, much work remains 
to be conducted on the implementation and empirical testing of recovery theory in future
research.

\appendix

\section{Proof of Theorem \ref{thm:monotone_usual}}
\label{app:pf_thm_II}

\begin{lemma}\label{lem:monotone}
	Assume $\delta<\lambda\leq\overline{\lambda}.$
	Let $g$ and $h$ be the functions corresponding to tuple $(\delta,M_\delta)$ and $(\lambda,M_\lambda),$ respectively. Then we have $g'g^{-1}>h'h^{-1}.$ 
\end{lemma}
\noindent For proof, see Lemma E.2 in \cite{park2016ross}.

\noindent   Now we prove Theorem \ref{thm:monotone_usual}.

\begin{proof}
	Assume that $h$ satisfies the usual conditions.
	From Lemma \ref{lem:monotone}, it is straightforward that $g'>0.$  We now show that 
	$\lim_{x\rightarrow\infty}g(x)=\infty\,.$ 
	It is enough to prove that $g(x)>h(x)$ for $x>\xi.$
	Integrating by $\int_{\xi}^{x}$ to the inequality $g'g^{-1}>h'h^{-1}$ in the Lemma \ref{lem:monotone}, we have $\ln g(x)-\ln g(\xi)>\ln h(x)-\ln h(\xi).$ Since $g(\xi)=h(\xi)=1,$ it follows that $g(x)>h(x)$ for $x>\xi.$
	In a similar way, one can prove that 
	$\lim_{x\rightarrow0+}g(x)=0$ by showing $g(x)<h(x)$ for $0<x<\xi.$
\end{proof}

\section{Proof of Theorem \ref{thm:usual_condi_iff_condi}}
\label{app:pf_thm_I}

The following lemmas will be used to prove Theorem \ref{thm:usual_condi_iff_condi} when $r(x)$ is  bounded near $0.$ Recall the definition
$\overline{r}:=\inf_{x>0} r(x).$ 

\begin{lemma}\label{lem:no_local_extreme}
	Assume that $r(\cdot)\geq 0$ and $\lambda<\overline{r}.$
	Let $h$ be a solution of $\mathcal{L}h=-\lambda h.$
	Then $h$ can attain nether a positive local maximum nor a negative local minimum. 	 
\end{lemma}

\begin{proof}
	Suppose that $h$ has a positive local maximum at $x_0.$ Then 	 
	$h'(x_0)=0$ and $h''(x_0)\leq0.$ From $\mathcal{L}h=-\lambda h,$ it follows that
	$$\frac{1}{2}\sigma^2(x_0){h''(x_0)}=\frac{1}{2}\sigma^2(x_0){h''(x_0)}+k(x_0)h'(x_0)=(r(x_0)-\lambda)h(x_0)>0\;,$$
	which is a contradiction. In a similar way, one can show that 
	$h$ cannot attain  a negative local minimum.
\end{proof}

\begin{lemma}\label{lem:order_b_c}
	Assume  $r(\cdot)\geq 0$ and $\lambda<\overline{r}.$
	Let $b$ and $c$ be two numbers with $c<b,$ not necessarily to be between $m_\lambda$ and $M_\lambda,$ and let $h_b$ and $h_c$ be the functions corresponding to tuples $(\lambda,b)$ and $(\lambda,c),$ respectively. Then $h_b(x)>h_c(x)$ for $x>\xi$ and $h_b(x)<h_c(x)$ for $0<x<\xi.$
\end{lemma}

\begin{proof}
	Since $h_c'(\xi)=c<b=h_b'(\xi)$ and $h_b(\xi)=h_c(\xi)=1,$ there exists an interval $(\xi,x_1)$ in which $h_c<h_b.$ Suppose $x_1<\infty.$ Then $h_b(x_1)=h_c(x_1).$ This means that $g:=h_b-b_c$ is a solution of $\mathcal{L}g=-\lambda g$ and $g$ has two zeros at $\xi$ and $x_1.$ By Lemma \ref{lem:no_local_extreme}, since $g$ can attain nether a positive local maximum nor a negative local minimum, $g$ should be identically zero. This leads us a contradiction.
	In a similar way, it can be shown that $h_b(x)<h_c(x)$ for $0<x<\xi.$
\end{proof}

\begin{lemma}\label{lem:slope}
	Assume that $r(\cdot)\geq 0$ and $\lambda<\overline{r}.$ Let  $h_{M(\lambda)}$ and $h_{m(\lambda)}$ be the functions corresponding to $(\lambda,M_\lambda)$ and $(\lambda,m_\lambda),$ respectively.  Then $h_{M(\lambda)}'>0$  and $h_{m(\lambda)}'<0.$ 
\end{lemma}

\begin{proof}
	We show that $h_{M(\lambda)}'>0.$    
	Let $b$ be a real number with $b>M_\lambda$ and denote by $h_b$  the function corresponding to $(\lambda,b).$
	We prove that  $h_b$ is a monotone increasing function. 
	By definition of $M_\lambda,$ $h_b$ has a zero at a point $x_0>0.$ It follows that $x_0<\xi$ because $h_{M(\lambda)}$ is positive and $h_b(x)>h_{M(\lambda)}(x)$ for $x>\xi$ by Lemma \ref{lem:order_b_c}. We have that  $h_b(x)$ is monotone increasing on $x>x_0$ since $h_b(x_0)=0$ and $h_b(\xi)=1,$ otherwise $h_b$ has a positive local maximum, which is a contradiction to Lemma \ref{lem:no_local_extreme}. Clearly, $h_b'(x_0)>0$ because if $h_b'(x_0)=0,$ $h_b$ is identically zero.
	$h_b$ is strictly increasing near $x_0,$ and thus $h_b$ is monotone increasing on $x<x_0$ since $h_b$ has no negative local minimum. In conclusion, $h_b$ is monotone increasing on $(0,\infty).$

	We now show that   	$h_{M(\lambda)}'(x)>0$ for all $x>0.$
	It can be easily shown that   
	$$h_{M(\lambda)}(x)=\lim_{b\rightarrow M(\lambda)+} h_b(x)$$   
	because $h_b$ can be expressed by the linear combination of 
	$h_{M(\lambda)}$ and $h_{M(\lambda)}:$
	$$h_b=\frac{b-m}{M -m}\,h_{M}+\frac{M-b}{M-m}\,h_{m}\;.$$
	Here, for a moment, we used $M$ and $m$ instead of $M(\lambda)$ and $m(\lambda),$ respectively, to avoid the heavy notions.
	Since $h_{M(\lambda)}$ is the limit of monotone increasing functions, $h_{M(\lambda)}$ is a monotone increasing function, that is,  
	$h_{M(\lambda)}'(x)\geq0$ for all $x>0.$
	Suppose that $h_{M(\lambda)}'(x_1)=0$ at some point $x_1>0.$ Then 
	$$\frac{1}{2}\sigma^2(x_1){h''(x_1)}=\frac{1}{2}\sigma^2(x_1){h''(x_1)}+k(x_1)h'(x_1)=(r(x_1)-\lambda)h(x_1)>0\;.$$
	Here, for a moment, we used $h$ instead of $h_{M(\lambda)}$ to avoid the heavy notions.  Thus, 
	$h_{M(\lambda)}''(x_1)>0,$ which  contradicts to the fact that $h_{M(\lambda)}$ is monotone increasing.
	In conclusion, it is obtained that 
	$h_{M(\lambda)}'(x)>0$ for all $x>0.$
	It can be shown that   $h_{m(\lambda)}'(x)<0$ for all $x>0$ in a similar way.
\end{proof}

We now prove Theorem \ref{thm:usual_condi_iff_condi}
\begin{proof}
	($\Leftarrow$) Suppose that $0$ is a natural boundary. Note that by Assumption \ref{assume:X}, $\infty$ is a natural or entrance boundary.
	Fix $\lambda<\overline{r}$ and let $h$ be the function corresponding to the tuple $(\lambda,M_\lambda).$ We show that $h$ satisfies the usual conditions. 
	By Lemma \ref{lem:slope}, it is obtained that $h'>0,$ which is one of the usual conditions. Now we show that $\lim_{x\rightarrow\infty}h(x)=\infty.$ Suppose that $\lim_{x\rightarrow\infty}h(x)$ is finite.
	Recall the function of $\gamma$ from Definition \ref{def:boundary_clasify}.
	By direct calculation, 	\begin{align}\label{eqn:a}
	\left(\frac{h'}{\gamma}\right)'=\frac{2(r-\lambda)h}{\sigma^2\gamma}\;.
	\end{align}
	It follows that   
	\begin{equation}\label{eqn:b}
	h'(x)=\gamma(x)\left(h'(\xi)+\int_{\xi}^{x}\frac{2(r(s)-\lambda)h(s)}{\sigma^2(s)\gamma(s)}\,ds \right)\;.
	\end{equation} 
	The two terms in the right-hand side are in $L^1(\xi,\infty)$ because they are positive and the integration of the left hand side $\int_\xi^\infty h'(y)\,dy=\lim_{x\rightarrow\infty}h(x)-h(\xi)$ is finite. Thus, $\gamma(x)\int_{\xi}^{x}\frac{2 h(s)}{\sigma^2(s)\gamma(s)}\,ds$ is in $L^1(\xi,\infty)$ since $r(x)-\lambda\geq\overline{r}-\lambda>0.$   On the other hand,
	since $h$ is increasing and $h(\xi)=1,$ we know
	$$R(x)=\gamma(x)\int_{\xi}^x\frac{2}{\sigma^2(s)\gamma(s)} \,ds\leq\gamma(x)\int_{\xi}^x\frac{2h(s)}{\sigma^2(s)\gamma(s)} \,ds\;.$$
	Therefore, $R\in L^1(\xi,\infty),$ which contracts the assumption that $\infty$ is inaccessible. 
	
	Now we prove that  $\lim_{x\rightarrow0+}h(x)=0.$
	Suppose $\lim_{x\rightarrow0+}h(x)>0.$ Since $h'/\gamma$ is increasing by equation \eqref{eqn:a} and $h'$ is positive, the limit exists
	\begin{equation}\label{eqn:aa}
	\lim_{x\rightarrow0+}\frac{h'(x)}{\gamma(x)}=C\geq 0\;. 
	\end{equation}
	Integrating the equation  \eqref{eqn:a} by $\int_0^x,$ we have
	\begin{equation}\label{eqn:ab}
	h'(x)=\gamma(x)\left(C+\int_{0}^{x}\frac{2(r(s)-\lambda)h(s)}{\sigma^2(s)\gamma(s)}\,ds \right)\;. 
	\end{equation}
	The two terms in the right-hand side are in $L^1(0,\xi)$ because they are nonnegative and $\int_0^\xi h'(x)\,dx=h(\xi)-\lim_{x\rightarrow0+}h(x)$ is finite.
	It follows that 
	$\gamma(x)\int_0^{x}\frac{2 h(s)}{\sigma^2(s)\gamma(s)}\,ds$ is in $L^1(0,\xi)$ since $r(x)-\lambda\geq\overline{r}-\lambda>0.$  Because $h$ is an increasing function and $\lim_{x\rightarrow0+}h(x)>0,$  we have that 
	$\gamma(x)\int_0^{x}\frac{2}{\sigma^2(s)\gamma(s)}\,ds$ is in $L^1(0,\xi).$ By the Fubini theorem, $$\int_0^\xi\gamma(x)\int_0^{x}\frac{2}{\sigma^2(s)\gamma(s)}\,ds\,dx
	=\int_0^\xi\frac{2}{\sigma^2(s)\gamma(s)}\int_s^{\xi}\gamma(x)\,dx\,ds$$
	is finite. Thus, $Q$ is in $L^1(0,\xi),$ which means that $0$ is not a natural boundary. This is a contradiction.
	
	($\Rightarrow$) Assume that the set $\mathcal{U}$ is nonempty. Let $(\lambda,M_{\lambda})$ be an element in the set and denote by $h$ the corresponding function. By Thorem \ref{thm:monotone_usual}, we may assume $\lambda<\overline{r}.$ We now prove that $0$ is a natural boundary. Suppose that $0$ is an entrance boundary, that is, $Q\in L^1(0,\xi).$  
	By the Fubini theorem, we know that
	\begin{equation}\label{eqn:ad} \gamma(x)\int_0^{x}\frac{2}{\sigma^2(s)\gamma(s)}\,ds \in L^1(0,\xi)\;.
	\end{equation}	
	Since $h$ is bounded on $(0,\xi)$ and $r$ is assumed to be bounded near $0,$
	it is obtained that  
	\begin{equation}\label{eqn:ac}
	\gamma(x)\int_0^{x}\frac{2(r(s)-\lambda)h(s)}{\sigma^2(s)\gamma(s)}\,ds \in L^1(0,\xi)\;.
	\end{equation}

	On the other hand, from equation \eqref{eqn:aa}, it follows that  
	$\lim_{x\rightarrow0+}\frac{h'(x)}{\gamma(x)}=C\geq 0\;. $
	We show that $C=0.$
	Suppose $C\neq0.$ Then, from equation \eqref{eqn:ab} and \eqref{eqn:ac},  $\gamma$ is in $L^1(0,\xi).$
	For $x<\xi/2,$ we know
	$$0\leq\frac{2}{\sigma^2(x)\gamma(x)}=\frac{2}{\sigma^2(x)\gamma(x)}\frac{\int_x^{\xi}\gamma(s)\,ds}{\int_x^{\xi}\gamma(s)\,ds}
	\leq\frac{2}{\sigma^2(x)\gamma(x)}\frac{\int_x^{\xi}\gamma(s)\,ds}{\int_{\xi/2}^{\xi}\gamma(s)\,ds}=\frac{-Q(x)}{\int_{\xi/2}^{\xi}\gamma(s)\,ds}\;,$$
	so we have that $c_1:=\int_0^\xi\frac{2}{\sigma^2(x)\gamma(x)}\,dx$ is finite  since $Q\in L^1(0,\xi).$
	It follows that $R\in L^1(0,\xi)$ because for $x<\xi,$
	$$0\leq-R(x)=\gamma(x)\int_x^{\xi}\frac{2}{\sigma^2(s)\gamma(s)} \,ds\leq\gamma(x)\int_0^{\xi}\frac{2}{\sigma^2(s)\gamma(s)} \,ds=c_1\gamma(x)\in L^1(0,\xi)\;.$$
	This is a contradiction because $\infty$ is an inaccessible boundary. Thus $C=0.$

	Now we prove that the hypothesis $Q\in L^1(0,\xi)$ induce a contradiction.
	Since $C=0$ and $h$ is increasing, from  equation \eqref{eqn:ab} and \eqref{eqn:ad}, we know
	\begin{align*} 
	\frac{h'(x)}{h(x)}
	&=\frac{\gamma(x)}{h(x)}\int_{0}^{x}\frac{2(r(s)-\lambda)h(s)}{\sigma^2(s)\gamma(s)}\,ds
	\leq{\gamma(x) }\int_{0}^{x}\frac{2(r(s)-\lambda)}{\sigma^2(s)\gamma(s)}\,ds
	\\
	&\leq{c_2\gamma(x) }\int_{0}^{x}\frac{2 }{\sigma^2(s)\gamma(s)}\,ds \in L^1(0,\xi)
	\end{align*}
	where $c_2:=|\lambda|+\sup_{0\leq x\leq\xi}r(x).$
	This implies that $\lim_{x\rightarrow0+}h(x)>0$ 
	because 
	$$-\lim_{x\rightarrow0+}\ln h(0)=\ln h(\xi)-\lim_{x\rightarrow0+}\ln h(0)=\int_0^\xi\frac{h'(x)}{h(x)}dx$$
	is finite. 
	This a contradiction to the assumption that $h$ satisfies the usual condition.
\end{proof}

\section{Proof of Theorem \ref{thm:usual_condi_const_r}}
\label{app:pf_thm_const_r}

Theorem \ref{thm:usual_condi_const_r} will be shown. For this purpose, we first prove
Theorem \ref{thm:iff_condi_increasing},
Lemma \ref{lemma:0_condition} and \ref{lemma:infty_condition} stated  below.
In this appendix, assume $r<\overline{\lambda}.$

\begin{thm}\label{thm:iff_condi_increasing}
	Let $r<\lambda\leq\overline{\lambda}.$ Denote by $h$ the function corresponding to $(\lambda,M_\lambda).$ 
	Then $\int_{\xi}^{\infty}\gamma(x)\,dx=\infty$ if and only if $h'>0.$ 
\end{thm}

\begin{proof} 
	Suppose that $\int_{\xi}^{\infty}\gamma(x)\,dx=\infty.$
	Recall from equation \eqref{eqn:b} that
	\begin{equation} \label{eqn:derivative_h_const_r}
	h'(x)=\gamma(x)\left(h'(\xi)-2(\lambda-r)\int_{\xi}^{x}\frac{h(s)}{\sigma^2(s)\gamma(s)}\,ds \right)\;.
	\end{equation} 
	If $h'(x_0)<0$ at some point $x_0,$ then for $x>x_0,$ 
	\begin{align*}
	h'(x)&=\gamma(x)\left(h'(\xi)-2(\lambda-r)\int_{\xi}^{x}\frac{h(s)}{\sigma^2(s)\gamma(s)}\,ds\right)\\
	&\leq\gamma(x)\left(h'(\xi)-2(\lambda-r)\int_{\xi}^{x_0}\frac{h(s)}{\sigma^2(s)\gamma(s)}\,ds\right)\\
	&=\frac{\gamma(x)}{\gamma(x_0)}\,h'(x_0)\;.
	\end{align*}
	Since  $\int_{\xi}^{\infty}\gamma(x)\,dx=\infty$ and $h'(x_0)<0,$ it follows that $h(x)\rightarrow-\infty$ as $x\rightarrow\infty,$ which is a contradiction to the fact that $h$ is positive. Therefore, we obtain that $h'\geq0.$
	We now show that $h'>0.$ Suppose that $h'(x_1)=0$ at some point. Then 
	$$\frac{1}{2}\sigma^2(x_1){h''(x_1)}=\frac{1}{2}\sigma^2(x_1){h''(x_1)}+k(x_1)h'(x_1)=-(\lambda-r)h(x_1)<0\;,$$	
	thus $h$ has local maximum at $x_1.$ This is a contradiction to the fact that $h'\geq0.$ 
	
	We now prove the converse. Assume that $h'>0.$   From equation \eqref{eqn:derivative_h_const_r}, we know
	\begin{equation*} 
	h'(\xi)\int_\xi^x\gamma(y)\,dy= h(x)-h(\xi)+2(\lambda-r)\int_\xi^x\gamma(y)\int_{\xi}^{y}\frac{h(s)}{\sigma^2(s)\gamma(s)}\,ds\,dy\;.
	\end{equation*} 
	To show that $\int_{\xi}^{\infty}\gamma(x)\,dx=\infty,$ since  the first term $h(x)$ of the right-hand side is positive, it is enough to show that 
	$$\int_\xi^\infty\gamma(x)\int_{\xi}^{x}\frac{h(s)}{\sigma^2(s)\gamma(s)}\,ds\,dx=\infty\;.$$
	On the other hand, because $\infty$ is an inaccessible boundary, we know  
	$$\int_\xi^\infty\gamma(x)\int_{\xi}^{x}\frac{1}{\sigma^2(s)\gamma(s)}\,ds\,dx=\infty\;.$$
	Since $h$ is positive and increasing, we obtain the desired result.
\end{proof}

\begin{lemma}\label{lemma:0_condition}
	Let $r<\lambda\leq\overline{\lambda}.$ If $h$ is a positive and increasing solution of $\mathcal{L}h=-\lambda h,$ then $h(0):=\lim_{x\rightarrow0+}h(x)=0.$
\end{lemma}
\begin{proof}
	It can be easily shown that $e^{(\lambda-r) t}h(X_t)$ is a local martingale. Since a positive local martingale is a supermartingale, it follows that 
	$$\mathbb{E}[e^{(\lambda-r) t}h(X_t)]\leq h(X_0)=h(\xi)\;.$$
	Since $h$ is increasing, we have
	$$e^{(\lambda-r) t}h(0)\leq \mathbb{E}[e^{(\lambda -r)t}h(X_t)]\leq h(\xi)\;.$$
	Thus,  $h(0)\leq e^{-(\lambda-r) t}h(\xi).$ Letting $t\rightarrow\infty,$ we obtain the desired result. 
\end{proof}

\begin{lemma}\label{lemma:infty_condition}
	Let $\lambda<\overline{\lambda}.$  	Let $h$  be the function corresponding to $(\lambda,M_\lambda),$ respectively. 
	If $(\lambda,h)$ satisfies the martingale condition and $\int_{\xi}^{\infty}\gamma(x)\,dx=\infty,$ then $h$ is unbounded.
\end{lemma}

\begin{proof}
	Let $h$ and $g$ be the functions corresponding to $(\lambda,M_\lambda)$ and $(\overline{\lambda},M_{\overline{\lambda}}).$
	Since  $(\lambda,h)$ satisfies the martingale condition, so does $(\overline{\lambda},M_{\overline{\lambda}})$ by Theorem \ref{thm:monotone_martingality}.
	Let $\mathbb{P}$ be the transformed measure with respect to $(\overline{\lambda},M_{\overline{\lambda}}).$ 
	Then 
	\begin{align}\label{eqn:unbounded}
	h(\xi)=\mathbb{E}^\mathbb{Q}[e^{({\lambda}-r)t}h(X_t)]=\mathbb{E}^\mathbb{P}[(g^{-1}h)(X_t)]\,e^{-(\overline{\lambda}-\lambda)t}g(\xi)\;.
	\end{align}
	The first equality holds because  $(\lambda,h)$ satisfies the martingale condition.
	Suppose that $h$ is bounded, that is, $h<c$ for some constant $c.$ By Lemma \ref{lem:monotone},
	as in the proof of Theorem  \ref{thm:monotone_usual}, we know that  $h(x)<g(x)$ for $0<x<\xi$ and  $h(x)>g(x)$ for $x>\xi,$ that is, $(g^{-1}h)(x)<1$ for $0<x<\xi$ and  $(g^{-1}h)(x)<cg^{-1}(x)$ for $x>\xi.$
	Since $\int_{\xi}^{\infty}\gamma(x)\,dx=\infty,$ $g$ is increasing by Theorem \ref{thm:iff_condi_increasing}, so we know that $(g^{-1}h)(x)<cg^{-1}(x)<cg^{-1}(\xi)=c$ for $x>\xi.$
	Thus, it is obtained that $(g^{-1}h)(x)<c+1$ for all $x>0.$
	From equation \eqref{eqn:unbounded}, we have
	\begin{align*} 
	h(\xi)=\mathbb{E}^\mathbb{P}[(g^{-1}h)(X_t)]\,e^{-(\overline{\lambda}-\lambda)t}g(\xi)\leq(c+1)e^{-(\overline{\lambda}-\lambda)t}g(\xi)\;.
	\end{align*}	
	Letting $t\rightarrow\infty,$ it follows that $h(\xi)\leq0.$
	This leads us a contradiction.
\end{proof}

$ $

\begin{proof} We now prove Theorem \ref{thm:usual_condi_const_r}.
	Suppose that $\int_{\xi}^{\infty}\gamma(x)\,dx=\infty.$  Fix any $\lambda$ with $r<\lambda<\overline{\lambda}$ and let $h$ be the function corresponding to $(\lambda,M_\lambda).$ By Theorem \ref{thm:iff_condi_increasing}, it follows that $h'>0.$ By Lemma \ref{lemma:0_condition}, since $h$ is positive and increasing, $\lim_{x\rightarrow0+}h(x)=0.$ By Lemma \ref{lemma:infty_condition}, we have $\lim_{x\rightarrow\infty}h(x)=\infty.$ Thus, $h$ satisfies the usual condition, that is, 
	$(\lambda,M_\lambda)$ is in $\mathcal{U}.$ By Theorem \ref{thm:monotone_usual}, it is obtained that
	$\{(\lambda,M_\lambda)\,|\,\lambda< \overline{\lambda}\,\}\subseteq\mathcal{U},$ which is the desired result.
	
	We now prove that if $\infty$ is a natural boundary, then $\mathcal{U}=\{(\lambda,M_\lambda)\,|\,\lambda\leq \overline{\lambda}\,\}.$ 
	Let $\overline{h}$ be the function corresponding to $(\overline{\lambda},M_{\overline{\lambda}}).$ 
	By the same argument, we have $\overline{h}\,'>0$ and 
	$\lim_{x\rightarrow0+}\overline{h}(x)=0.$
	It suffices to show that 
	$\lim_{x\rightarrow\infty}\overline{h}(x)=\infty.$ 
	From equation \eqref{eqn:b}, we have
	\begin{equation*} 
	\overline{h}\,'(x)=\gamma(x)\left(\overline{h}'(\xi)-2(\overline{\lambda}-r)\int_{\xi}^{x}\frac{\overline{h}(s)}{\sigma^2(s)\gamma(s)}\,ds \right)\,.
	\end{equation*} 
	Because the term inside the parenthesis is a decreasing function of $x$ and the left-hand side $\overline{h}\,'(x)$ is positive for all $x,$ by letting $x\rightarrow\infty,$ it follows that $C:=\overline{h}\,'(\xi)-2(\overline{\lambda}-r)\int_{\xi}^{\infty}\frac{\overline{h}(s)}{\sigma^2(s)\gamma(s)}\,ds$ is nonnegative. The above equation can be written by
	\begin{equation}\label{eqn:o} 
	\overline{h}\,'(x)=\gamma(x)\left(C+2(\overline{\lambda}-r)\int_{x}^{\infty}\frac{\overline{h}(s)}{\sigma^2(s)\gamma(s)}\,ds \right)\geq2(\overline{\lambda}-r)\gamma(x)\int_{x}^{\infty}\frac{\overline{h}(s)}{\sigma^2(s)\gamma(s)}\,ds\,.
	\end{equation} 
	On the other hand, applying the Fubini theorem, we know that 
	$$\int_\xi^\infty\gamma(x)\int_x^{\infty}\frac{2}{\sigma^2(s)\gamma(s)}\,ds\,dx
	=\int_\xi^\infty\frac{2}{\sigma^2(s)\gamma(s)}\int_\xi^{s}\gamma(x)\,dx\,ds=\infty $$
	because $\infty$ is a natural boundary. Thus, 
	the integration of the right hand side of equation \eqref{eqn:o} becomes
	$\int_\xi^\infty\gamma(x)\int_x^{\infty}\frac{\overline{h}(s)}{\sigma^2(s)\gamma(s)}\,ds\,dx=\infty$ since $\overline{h}$ is an increasing function. This implies that $\lim_{x\rightarrow\infty}\overline{h}(x)=\infty.$  
	
	Now suppose that $\int_{\xi}^{\infty}\gamma(x)\,dx<\infty.$
	By Theorem \ref{thm:iff_condi_increasing},
	for any $\lambda$ with $r<\lambda\leq \overline{\lambda},$
	the corresponding function $h$ does not satisfy the  condition $h'>0.$ Thus, $\mathcal{U}\subseteq\{(\lambda,M_\lambda)\,|\,\lambda\leq r\,\}.$ When $\lambda=r,$ by the same argument in the proof of Proposition \ref{prop:U_when_overline_lambda_r}, we know that $(r,M_r)$  is not in $\mathcal{U},$ so $\mathcal{U}\subseteq\{(\lambda,M_\lambda)\,|\,\lambda<r\,\}.$
	On the other hand, $\{(\lambda,M_\lambda)\,|\,\lambda<r\,\}\subseteq\mathcal{U}$
	is clear by Theorem \ref{thm:usual_condi_iff_condi}. Therefore, 
	$\mathcal{U}=\{(\lambda,M_\lambda)\,|\,\lambda< r\,\}.$
\end{proof}

\section{An invariant property}
\label{sec:inv_prop}
Let $X_t$ be a diffusion process satisfying 
$$dX_t=k(X_t)\,dt+\sigma(X_t)\,dW_t$$
with the killing rate $r(x).$
Then the corresponding generator is 
$$\mathcal{L}h(x)=\frac{1}{2}\sigma^2(x){h''(x)}+k(x)h'(x)
-r(x)h(x)\;.$$
Fix an admissible pair $(\lambda,h)$ of the generator $\mathcal{L}.$

Let $(a,b)$ an open interval in $\mathbb{R}.$ Suppose that  $\pi:(0,\infty)\mapsto(a,b) $ is a continuously twice differentiable bijective map with continuously twice differentiable inverse. Then $\pi$ is  increasing or decreasing, so we may assume that $\pi$ is increasing.
Define $Y_t=\pi(X_t)$ and $H(y)=h(\pi^{-1}(y)).$
Then $Y_t$ satisfies
\begin{equation*}
\begin{aligned}
dY_t
&=(k\pi'+\frac{1}{2}\sigma^2\pi'')(X_t)\,dt+(\sigma\pi')(X_t)\,dW_t\\
&=(k\pi'+\frac{1}{2}\sigma^2\pi'')\circ(\pi^{-1}(Y_t))\,dt+(\sigma\pi')\circ(\pi^{-1}(Y_t))\,dW_t
\end{aligned}
\end{equation*}
with the killing rate $r(\pi^{-1}(y)).$
The corresponding generator is
$$\mathcal{L}^\pi H(y)=\frac{1}{2}(\sigma\pi')^2\circ(\pi^{-1}(y)){H''(y)}+(k\pi'+\frac{1}{2}\sigma^2\pi'')\circ(\pi^{-1}(y))H'(y)
-r(\pi^{-1}(y))H(y)\;.$$
The pair $(\lambda,H)$ is a solution pair of $\mathcal{L}^\pi H=-\lambda H.$ In addition, we have
$$\left.\frac{d\mathbb{P}_h}{d\mathbb{Q}\,}\right|_{\mathcal{F}_{t}}=e^{\lambda t}\,h(X_{t})\,G_t^{-1}=e^{\lambda t}H(Y_t)\,G_t^{-1}=\left.\frac{d\mathbb{P}_H}{d\mathbb{Q}\,}\right|_{\mathcal{F}_{t}}\;,$$
thus $(\lambda,h)$ and $(\lambda,H)$ induce the same transformed measures.
Since $X_t$ approaches to infinity as $t\rightarrow\infty$ under $\mathbb{P}_h$ (by the definition of an admissible pair), the process $Y_t=\pi(X_t)$  
approaches to the right boundary $b$ as $t\rightarrow\infty$ under $\mathbb{P}_H.$
It is clear that the function $H$ satisfies the usual condition (i),(ii)
and (iii), (iv) replaced the limits $$\lim_{x\rightarrow0+}h(x)=0\;,\;\lim_{x\rightarrow\infty}h(x)=\infty$$
by
$$\lim_{y\rightarrow a+}H(y)=0\;,\;\lim_{y\rightarrow b-}H(y)=\infty\;.$$

\bibliographystyle{plainnat}

\bibliography{do_prices}

\end{document}